\documentclass[journal, draftclsnofoot,oneside, onecolumn,11pt, letterpaper]{IEEEtran}

\ifCLASSINFOpdf
\else
\fi
\usepackage{pslatex}
\usepackage{amsfonts,color,morefloats}
\usepackage{amssymb,amsmath,latexsym,amsthm}
\usepackage{bbm}
\usepackage{amsfonts}
\usepackage{mathrsfs}
\usepackage{amsthm}
\usepackage{latexsym}
\usepackage{color}
\usepackage{dcolumn}
\usepackage{extarrows}

\newtheorem{theorem}{Theorem}
\newtheorem{lemma}[theorem]{Lemma}

\newtheorem{definition}{Definition}

\begin{document}
%
\title{ The 4-Adic Complexity of A Class of Quaternary Cyclotomic Sequences with\\ Period $2p$}
%
%
%

\author{Shiyuan~Qiang,
        Yan~Li,
        Minghui~Yang*,
        and~Keqin~Feng
\thanks{*Corresponding author}
\thanks{The work was supported by the National Science Foundation of China (NSFC) under Grant 12031011, 11701553.}
\thanks{Shiyuan Qiang is with the department of Applied Mathematics, China Agricultural, university, Beijing 100083, China (e-mail: qsycau\_18@163.com).}
\thanks{Yan Li is with the department of Applied Mathematics, China Agricultural, university, Beijing 100083, China (e-mail: liyan\_00@cau.edu.cn).}
\thanks{Minghui Yang is with State Key Laboratory of Information Security, Institute of Information Engineering, Chinese Academy of Sciences, Beijing 100093, China (e-mail:  yangminghui6688@163.com).}
\thanks{Keqin Feng is with the department of Mathematical Sciences, Tsinghua University, Beijing 100084, China (email: fengkq@mail.tsinghua.edu.cn).}
}

%
%

\markboth{}%
{Shell \MakeLowercase{\textit{et al.}}: Bare Demo of IEEEtran.cls for IEEE Journals}
%



\maketitle

\begin{abstract}
In cryptography, we hope a sequence over $\mathbb{Z}_m$ with period $N$ having larger $m$-adic complexity.Compared with the binary case, the computation of 4-adic complexity of knowing quaternary sequences has not been well developed. In this paper, we determine the 4-adic complexity of the quaternary cyclotomic sequences with period 2$p$ defined in \cite{Y.-J. Kim}. The main method we utilized is a quadratic Gauss sum $G_{p}$ valued in $\mathbb{Z}_{4^N-1}$ which can be seen as a version of classical quadratic Gauss sum. Our results show that the 4-adic complexity of this class of quaternary cyclotomic sequences reaches the maximum if $5\nmid p-2$ and close to the maximum otherwise.

\end{abstract}

\begin{IEEEkeywords}
4-adic complexity, quaternary cyclotomic sequences, quadratic Gauss sum, cryptography
\end{IEEEkeywords}

%
\IEEEpeerreviewmaketitle

\section{Introduction}

Periodic sequences over finite field $\mathbb{F}_q$ or finite ring $\mathbb{Z}_m=\mathbb{Z}/m\mathbb{Z}$ have many important applications in spread-spectrum multiple-access communication and cryptography. In the stream cipher schemes we need the sequences having good pseudorandom cryptographic properties and large linear complexity \cite{T. W. Cusick}.

In the past three decades, many series of such sequences have been investigated, their autocorrelation and linear complexity have been determined or estimated. A sequence with linear complexity $n$ can be generated by a linear shift register of length $n$ and the period of a sequence is an upper bound of $n$. In 1990's, Klapper, Goresky and Xu [4, 5] described a kind of non-linear shift registers (feedback with carry shift registers (FCSRs)) and raised a new complexity, called $m$-adic complexity.

\begin{definition} \label{def1} Let $m\geq 2$ be a positive integer, $A=\{a(i)\}_{i=0}^{N-1}$ be a sequence over $\mathbb{Z}_m$ with period $N$, $a(i)\in\{0, 1, \ldots, m-1\}$ for $0\leq i \leq N-1$. Let $S_{A}(m)=\sum_{i=0}^{N-1}a(i)m^{i}\in \mathbb{Z}$ and $d=\gcd(S_{A}(m), m^{N}-1)$. The m-adic complexity of the sequence A is defined by
$$C_A(m)=\log_m \left(\frac{m^N-1}{d}\right).$$
\end{definition}

Roughly speaking, a sequence $A$ with period $N$ over $\mathbb{Z}_m$ can be generated by an FCSR  of length $\lceil C_A(m)\rceil$. In cryptography, we hope a sequence $A$ over $\mathbb{Z}_m$ with period $N$ having larger $m$-adic complexity $C_A(m)$. By the Definition \ref{def1} we know that $\lceil C_A(m) \rceil \leq N $, where for $\alpha >0$, $\lceil \alpha \rceil$ is the smallest integer $n$ such that $n\geq \alpha$.

In the past decade, the 2-adic complexity $C_A(2)$ has been determined or estimated for many binary sequences $A$ with good autocorrelation properties. Particularly, for all known binary sequences $A=\{a(i)\}_{i=0}^{N-1}\ (a(i)\in\{0,1\})$ with period $N \equiv 3\pmod 4$ and ideal autocorrelation $(\sum_{i=0}^{N-1}(-1)^{a(i+\tau)-a(i)}=-1$, for all $1\leq\tau\leq N-1)$, the 2-adic complexity $C_A(2)$ reaches the maximum value $\log_{2}(2^{N}-1)$ \cite{H. G. Hu}. On the other hand, quaternary sequences (over $\mathbb{Z}_4$) are also important sequences in many practical applications \cite{S. M. Krone}. Comparing with the binary case, the computation of 4-adic complexity of knowing quaternary sequences has not been well developed. In this paper we determine the 4-adic complexity of the quaternary cyclotomic sequences give by Kim et al. in \cite{Y.-J. Kim}.

Let $p$ be an odd prime, $(\frac{{.}}{p}): \mathbb{Z}_p^{\ast}=\{1, 2, \ldots, p-1\}\rightarrow \{\pm 1\}$ be the Legendre symbol. Namely, for $a\in \mathbb{Z}_p^{\ast}$,
\begin{align*}
 \big(\frac{a}{p}\big)
 = \left\{ \begin{array}{ll}
1, & \textrm{if $a$ is a square in $\mathbb{Z}_p^{\ast}$;}\\
-1, & \textrm{otherwise.}
\end{array} \right.
\end{align*}

Let $g$ be a primitive element modulo 2$p$, $\mathbb{Z}_{2p}^{\ast}=\langle g\rangle$ and
$$D_{0}^{(2p)}=\langle g^{2}\rangle\subseteq \mathbb{Z}_{2p}^{\ast}, \quad D_{1}^{(2p)}= gD_{0}^{(2p)}\subseteq \mathbb{Z}_{2p}^{\ast}$$
$$D_{0}^{(p)}=\{a\in \mathbb{Z}_p^{\ast}:\big(\frac{a}{p}\big)=1\}, \quad D_{1}^{(p)}=\{a\in \mathbb{Z}_p^{\ast}:\big(\frac{a}{p}\big)=-1\}$$
Then
$$\mathbb{Z}_{2p}=D_{0}^{(2p)}\bigcup D_{1}^{(2p)}\bigcup 2D_{0}^{(p)}\bigcup 2D_{1}^{(p)}\bigcup \{0, p\}\qquad(disjoint)$$
\begin{definition} \label{def2}(\cite{Y.-J. Kim}) Define a quaternary sequence $A=\{a(i)\}_{i=0}^{2p-1}$ over $\mathbb{Z}_{4}=\{0,1,2,3\}$ with period $N=2p$ by
\begin{align*}
 a(i)
 = \left\{ \begin{array}{ll}
0, & \textrm{if $i=0$ or $i \in D_{0}^{(2p)}$}\\
2, & \textrm{if $i=p$ or $i\in 2D_{0}^{(p)}$}\\
1, & \textrm{if $i\in D_{1}^{(2p)}$}\\
3, & \textrm{if $i\in 2D_{1}^{(p)}$}
\end{array} \right.
\end{align*}

\end{definition}

The autocorrelation of this quaternary sequence has been computed in \cite{Y.-J. Kim}. Du and Chen \cite{X. Du} translated this sequence into a sequence $A'$ over the finite field $\mathbb{F}_4$ by the Gray mapping and computed the linear complexity of $A'$ over $\mathbb{F}_4$. The following theorem is our main result which determines the 4-adic complexity of the quaternary sequence $A$.

\begin{theorem}\label{th1}
Let $A$ be the quaternary sequence with period $N=2p$ defined by Definition \ref{def2}. Then the 4-adic complexity of $A$ is
\begin{align*}
 C_{A}(4)
 = \left\{ \begin{array}{ll}
\log_{4}(\frac{4^{N}-1}{5}), & \textrm{if $5\mid p-2$;}\\
\log_{4}(4^{N}-1), & \textrm{otherwise.}
\end{array} \right.
\end{align*}
\end{theorem}

In Section \ref{sec2} we introduce a quadratic Gauss sum $G_{p}$ valued in $ \mathbb{Z}_{4^{N}-1}$ as a version of classical quadratic Gauss sum, prove a property of $G_{p}$, and show that $S_{A}(4)=\sum_{i=0}^{N-1}a(i)4^{i}$ can be expressed by $G_{p}$ modulo $4^{N}-1$. In Section \ref{sec3} we prove Theorem \ref{th1}.

\section{Preliminaries}\label{sec2}

Let $p$ be an odd prime, $N=2p$. From the fact that $a\equiv b\pmod p$ implies $4^{2a}\equiv 4^{2b}\pmod {4^{N}-1}$ we can define an element $G_{p}$ in $\mathbb{Z}_{4^{N}-1}$:

$$G_{p}=\sum_{a\in \mathbb{Z}_p^{\ast}}\big(\frac{a}{p}\big)4^{2a}=\sum_{a=1}^{p-1}(\frac{a}{p}\big)4^{2a}\pmod {4^{N}-1}$$
The following result shows that $S_{A}(4)$ can be expressed by $G_{p}$ modulo $4^{N}-1$ and $G_{p}$ has a similar property as classical quadratic Gauss sum.

\begin{lemma}\label{lem2}
Let $A=\{a(i)\}_{i=0}^{N-1}$ be the quaternary sequence over $\mathbb{Z}_{4}$ with period $N=2p$ $(p\geq 3)$ defined by Definition \ref{def2}.
Then
 \begin{align*}
\begin{array}{ll}
& \textrm{$\rm(1)$}\ S_{A}(4)\equiv \frac{1}{2}(3\cdot 4^{p}-5)+\frac{4^{p}+5}{2}\frac{4^{N}-1}{15}-\frac{1}{2}(\big(\frac{2}{p}\big)4^{p}+1)G_{p}\pmod {4^{N}-1}\\
& \textrm{$\rm(2)$}\ G_{p}^{2}\equiv \big(\frac{-1}{p}\big)(p-\frac{4^{N}-1}{15})\pmod {4^{N}-1}
\end{array}
\end{align*}
\end{lemma}

\begin{proof}
(1). By the Chinese Remainder Theorem, We have isomorphism of rings
$$\varphi: \mathbb{Z}_{2p}\cong \mathbb{Z}_{p}\oplus \mathbb{Z}_{2}$$
by $\varphi(x\pmod{2p})=(x\pmod p, x\pmod 2)$. It is easy to see that for any element $(A, B)\in \mathbb{Z}_{p}\oplus \mathbb{Z}_{2}  (0\leq A \leq p-1, B\in \{0,1\})$, $\varphi^{-1}(A, B)=A(p+1)+pB\in \mathbb{Z}_{2p}.$ Then
\begin{align*}
\sum_{i\in D_{1}^{(2p)}}4^{i}
\equiv \sum_{A=1\atop (\frac{A}{p})=-1}^{p-1}4^{A(p+1)+p}\equiv\sum_{A=1\atop (\frac{2A}{p})=-1}^{p-1}4^{2A+p}\pmod{4^{2p}-1}  \qquad(i=A(p+1)+p)
\end{align*}

From Definition \ref{def2} we know that
\begin{align*}
S_{A}(4)& =\sum_{i\in D_{1}^{(2p)}}4^{i}+2\cdot 4^{p}+2\sum_{a\in D_{0}^{(p)}}4^{2a}+3\sum_{a\in D_{1}^{(p)}}4^{2a}\\
&\equiv  \sum_{a=1\atop (\frac{2a}{p})=-1}^{p-1}4^{p+2a}+2\cdot 4^{p}+2\sum_{a=1\atop (\frac{a}{p})=1}^{p-1}4^{2a}+3\sum_{a=1\atop (\frac{a}{p})=-1}^{p-1}4^{2a}\pmod {4^{N}-1}\notag\\
&\equiv  4^{p}\cdot\sum_{a=1\atop (\frac{a}{p})=-(\frac{2}{p})}^{p-1}4^{2a}+2\cdot 4^{p}+2\sum_{a=1}^{p-1}4^{2a}+\sum_{a=1\atop (\frac{a}{p})=-1}^{p-1}4^{2a}\pmod {4^{N}-1}\notag\\
&\equiv  4^{p}\cdot \frac{1}{2}\sum_{a=1}^{p-1}(1-(\frac{2}{p})(\frac{a}{p}))4^{2a}+2\cdot 4^{p}+2\sum_{a=1}^{p-1}4^{2a}+\frac{1}{2}\sum_{a=1}^{p-1}(1-(\frac{a}{p}))4^{2a}\pmod {4^{N}-1}\notag\\
&\equiv(\frac{4^{p}}{2}+2+\frac{1}{2})\sum_{a=1}^{p-1}4^{2a}+2\cdot 4^{p}-\frac{1}{2}((\frac{2}{p})4^{p}+1)G_{p}\pmod {4^{N}-1}\notag\\
&\equiv \frac{4^{p}+5}{2}(\frac{4^{N}-1}{15}-1)+2\cdot 4^{p}-\frac{1}{2}((\frac{2}{p})4^{p}+1)G_{p}\pmod {4^{N}-1}\notag\\
&\equiv \frac{1}{2}(3\cdot 4^{p}-5)+ \frac{4^{p}+5}{2}\cdot \frac{4^{N}-1}{15}
-\frac{1}{2}((\frac{2}{p})4^{p}+1)G_{p}\pmod {4^{N}-1}\notag
\end{align*}
(2). By the definition of $G_{p}$,
\begin{align*}
G_{p}^{2}& =\sum_{x,y=1}^{p-1}\big(\frac{xy}{p}\big)16^{x+y}\\
&\equiv  \sum_{x,t=1}^{p-1}\big(\frac{t}{p}\big)16^{x(1+t)}\pmod {4^{N}-1} \qquad(y=xt)\notag\\
&\equiv \big(\frac{-1}{p}\big)(p-1)+\sum_{t=1}^{p-2}\big(\frac{t}{p}\big)\sum_{x=1}^{p-1}16^{x(1+t)}\pmod {4^{N}-1}\notag\\
&\equiv \big(\frac{-1}{p}\big)(p-1)+\sum_{t=1}^{p-2}\big(\frac{t}{p}\big)\sum_{x=1}^{p-1}16^{x}\pmod {4^{N}-1}\notag\\
&\equiv \big(\frac{-1}{p}\big)(p-1)-\big(\frac{-1}{p}\big)(\frac{16^{p}-1}{15}-1)\pmod {4^{N}-1}\notag\\
&\equiv \big(\frac{-1}{p}\big)(p-\frac{4^{N}-1}{15})\pmod {4^{N}-1}\notag
\end{align*}
\end{proof}

\section{Proof of Theorem \ref{th1}}\label{sec3}

By Definition \ref{def1}, $C_{A}(4)=\log_{4}(\frac{4^{N}-1}{d})$ where $d=\gcd(S_{A}(4), 4^{N}-1)$. Since $N=2p$, $4^{N}-1=(4^{p}+1)(4^{p}-1)$ and $\gcd(4^{p}+1, 4^{p}-1)=\gcd(4^{p}+1, 2)=1$. We get $d=d_{+}d_{-}$ where both of
$$d_{+}=\gcd(S_{A}(4), 4^{p}+1)\quad \rm{and} \quad d_{-}=\gcd(S_{A}(4), 4^{p}-1)$$
are odd. We need to determine $d_{+}$ and $d_{-}$.

\begin{lemma}\label{lem3}
\begin{align*}
 d_{+}
 = \left\{ \begin{array}{ll}
5, & \textrm{if $5\mid p-2$;}\\
1, & \textrm{otherwise.}
\end{array} \right.
\end{align*}
\end{lemma}

\begin{proof}
Let $\ell$ be a prime divisor of $d_{+}$. Then $S_{A}(4)\equiv 4^{p}+1\equiv 0\pmod\ell.$ From Lemma \ref{lem2}(1) and $\ell | 4^{p}+1$ we have
\begin{align}\label{L1}
S_{A}(4)\equiv -4+ \frac{2(4^{p}-1)}{15}(4^{p}+1)-\frac{1}{2}\big(-(\frac{2}{p})+1\big)G_{p}\pmod\ell
\end{align}
Since $4^{p}+1\equiv 2\pmod 3$, we know that $\ell \geq 5$. Firstly we consider the case $\ell =5$. In this case, $4^{p}+1\equiv (-1)^{p}+1\equiv 0\pmod 5$ and
\begin{align*}
 G_{p}
 = \sum_{a=1}^{p-1}\big(\frac{a}{p}\big)4^{2a}\equiv \sum_{a=1}^{p-1}\big(\frac{a}{p}\big)\equiv 0 \pmod 5
\end{align*}
Then by the formula (\ref{L1}) we get $S_{A}(4)\equiv -4- \frac{4}{15}(4^{p}+1)\pmod 5$. Therefore,
\begin{align*}
S_{A}(4)&\equiv 0\pmod 5 \Leftrightarrow \frac{4^{p}+1}{15}\equiv -1\pmod 5\Leftrightarrow 4^{p}+1\equiv-15\pmod {25}\\
&\Leftrightarrow 4^{p-2}\equiv -1\pmod {25}\Leftrightarrow 10\mid 2(p-2)\notag
\end{align*}
The last equivalence can be obtained by the fact  that the order of 4 modulo 25 is 10. Therefore, $5|d_{+}$ if and only if $5|p-2$. Moreover, assume that $5\mid p-2$. If $25\mid 4^{p}+1$, then $4^{2p}\equiv1\pmod {25}$ and then $10\mid 2p$ which contradicts to $5\mid p-2$. In summary, $5\mid d_{+}$ if and only if $5\mid p-2$ and when $5|p-2$ we have $25\nmid d_{+}.$

Now we assume $\ell \geq7$. The formula (\ref{L1}) becomes
\begin{align}
0\equiv S_{A}(4)\equiv -4- \frac{1}{2}\big(-(\frac{2}{p})+1\big)G_{p}\pmod\ell
\end{align}
and by Lemma \ref{lem2}(2), $G_{p}^{2}\equiv \big(\frac{-1}{p}\big)p\pmod \ell. $ If $p\equiv \pm1\pmod 8$ then $\big(\frac{2}{p}\big)=1$ and we get a contradiction $0\equiv -4\pmod\ell$. If $p\equiv \pm3\pmod 8$ then $\big(\frac{2}{p}\big)=-1$ and $G_{p}\equiv -4\pmod\ell$ by (2). Therefore $16\equiv G_{p}^{2}\equiv \big(\frac{-1}{p}\big)p\pmod\ell$ and then $\ell \mid p-16$ or $\ell \mid p+16$.
On the other hand, $2^{2p}=4^{p}\equiv -1\pmod\ell$ which means that the order of 2 modulo $\ell$ is $4p$. Therefore $4p\mid \ell-1$ and $4p\leq \ell-1\leq p+15$ which implies $p\leq 5$. If $p=3$, then $l=13$. From $G_3=\sum_{a=1}^2(\frac{a}{3})4^{2a}$ we get $G_3=4^2-4^4\equiv -6 \pmod {13}$. By (2) we get $G_3\equiv -4\pmod {13}$ which is a contradiction. Then by $p\equiv \pm3\pmod 8$ we get $p=5$  and $\ell=21$ which contradicts to that $\ell$ is a prime number. In summary, we get $d_{+}=5$ if $5\mid p-2$ and $d_{+}=1$ otherwise. This completes the proof of Lemma \ref{lem3}.
\end{proof}

\begin{lemma}\label{lem4}
$d_{-}=1$
\end{lemma}
\begin{proof}
Let $\ell$ be a prime divisor of $d_{-}$. Then $S_{A}(4)\equiv 4^{p}-1\equiv 0\pmod\ell.$ From
$$S_{A}(4)\equiv 2\cdot 4^{p}+\sum_{a\in D_{1}^{(2p)}}1+\sum_{a\in D_{0}^{(p)}}2\equiv 3\cdot\frac{p-1}{2}+2\equiv 2\pmod 3$$
we know that $\ell\geq 5$. From $4^{p}\equiv 1\pmod\ell$ and $\ell\geq 5$ we know that the order of 4 modulo $\ell$ is $p$. Therefore $p\mid \ell-1$. On the other hand, by Lemma \ref{lem2} we have
$$0\equiv S_{A}(4)\equiv -1-\frac{1}{2}(\big(\frac{2}{p}\big)+1)G_{p}\pmod l, \quad G_{p}^{2}\equiv \big(\frac{-1}{p}\big)p \pmod\ell$$
If $\big(\frac{2}{p}\big)=-1$, we get contradiction $0\equiv -1\pmod\ell$. If $\big(\frac{2}{p}\big)=1$ then $1\equiv -G_{p}\pmod l$ and $1\equiv G_{p}^{2}\equiv \big(\frac{-1}{p}\big)p\pmod\ell$. We get $\ell\mid\frac{p+1}{2}$ or $\ell\mid\frac{p-1}{2}$. Then we have $p\leq\ell-1\leq\frac{1}{2}(p+1)-1=\frac{1}{2}p-\frac{1}{2}$ which is a contradiction. Therefore we get $d_{-}=1$.
\end{proof}

{\bfseries{Proof of Theorem \ref{th1}}} \quad By Lemma \ref{lem3} and Lemma \ref{lem4} we get
\begin{align*}
 d=d_{+}d_{1}
 = \left\{ \begin{array}{ll}
5, & \textrm{if $5\mid p-2$;}\\
1, & \textrm{otherwise.}
\end{array} \right.
\end{align*}
Therefore
\begin{align*}
C_{A}(4)=\log_{4}(\frac{4^{N}-1}{d})
 = \left\{ \begin{array}{ll}
\log_{4}(\frac{4^{N}-1}{5}), & \textrm{if $5\mid p-2$;}\\
\log_{4}(4^{N}-1), & \textrm{otherwise.}
\end{array} \right.
\end{align*}


\begin{thebibliography}{1}

\bibitem{T. W. Cusick}
T. W. Cusick, C. Ding and A. Renvall, \emph{Stream Ciphers and Number Theory}, (revised edition), Elsevier, 2004.
\bibitem{X. Du}
X. Du and Z. Chen, ``Linear complexity of quaternary sequences generated using generalized cyclotomic classes modulo $2p$," IEICE Trans. Fundamentals, vol. E94-A, no. 5, pp. 1214-1217, 2011.
\bibitem{H. G. Hu}
H. G. Hu, ``Comments on a new method to compute the 2-adic complexity of binary sequences ," IEEE Trans. Inform. Theory, vol. 60, no. 9, pp. 5803-5804, 2014.
\bibitem{A. Klapper}
A. Klapper and M. Goresky, ``Feedback shift registers, 2-adic span, and combiners with memory," Jour of Cryptology, vol. 10, no. 2, pp. 111-147, 1997.
\bibitem{J. Xu}
A. Klapper and J. Xu, ``Algebraic feedback shift registers," Theoretical Computer Science, vol. 226, no. 1-2, pp. 61-92, 1999.
\bibitem{Y.-J. Kim}
Y.-J. Kim, Y.-P. Hong and H.-Y. Song, ``Autocorrelation of some quaternary cyclotomic sequences of length $2p$," IEICE Trans. Fundamentals, vol. E91-A, no. 12,  pp. 3679-3684, 2008.
\bibitem{S. M. Krone}
S. M. Krone and D. V. Sarwate, ``Quadriphase sequences for spread-spectrum multiple-access communication," IEEE Trans. Inform. Theory, vol. 30, no. 3, pp. 520-529, 1984.


\end{thebibliography}
\end{document}